\documentclass[conference]{IEEEtran}
\IEEEoverridecommandlockouts
\usepackage{amssymb}
\usepackage[cmex10]{amsmath}
\usepackage{stfloats}
\usepackage{graphicx}
\usepackage{subfigure}
\usepackage{tabularx}
\usepackage{epsfig,epsf,color,balance,cite}
\usepackage{verbatim}
\usepackage{url}
\usepackage{bm}

\usepackage{mathrsfs}
\usepackage{amsthm}
\usepackage{mathtools}
\usepackage{extarrows}
\DeclarePairedDelimiter{\ceil}{\lceil}{\rceil}
\usepackage[a4paper,left=1.2cm,right=1.2cm, top=1.2cm,bottom=1.80cm,textwidth=13.8cm,textheight=20.3cm,
includehead,headheight=0pt,headsep=0pt]{geometry}

\usepackage{algpseudocode}
\usepackage{graphicx}
\usepackage{epstopdf}
\usepackage{mathbbol}

\def\mindex#1{\index{#1}}

\def\sq{\hbox{\rlap{$\sqcap$}$\sqcup$}}
\def\qed{\ifmmode\sq\else{\unskip\nobreak\hfil
\penalty50\hskip1em\null\nobreak\hfil\sq
\parfillskip=0pt\finalhyphendemerits=0\endgraf}\fi\medskip}

\long\def\defbox#1{\framebox[.9\hsize][c]{\parbox{.85\hsize}{%
\parindent=0pt
\baselineskip=12pt plus .1pt      % STYLE
\parskip=6pt plus 1.5pt minus 1pt % CHANGES
 #1}}}

\long\def\beginbox#1\endbox{\subsection*{}%
\hbox{\hspace{.05\hsize}\defbox{\medskip#1\bigskip}}%
\subsection*{}}

\def\endbox{}

\newsavebox{\junk}
\savebox{\junk}[1.6mm]{\hbox{$|\!|\!|$}}

\def\det{{\mathop{\rm det}}}

\def\bfmath#1{{\mathchoice{\mbox{\boldmath$#1$}}%
{\mbox{\boldmath$#1$}}%
{\mbox{\boldmath$\scriptstyle#1$}}%
{\mbox{\boldmath$\scriptscriptstyle#1$}}}}

\def\bfmY{\bfmath{Y}}

\def\bfmhhaY{\bfmath{\hhaY}} %\widehat{\widehat{Y}}}}
\def\bfmhhaY{\hbox to 0pt{$\widehat{\bfmY}$\hss}\widehat{\phantom{\raise 1.25pt\hbox{$\bfmY$}}}}

\def\til={{\widetilde =}}

 \def\FRAC#1#2#3{\genfrac{}{}{}{#1}{#2}{#3}}

\def\ddtp{{\mathchoice{\FRAC{1}{d^{\hbox to 2pt{\rm\tiny +\hss}}}{dt}}%
{\FRAC{1}{d^{\hbox to 2pt{\rm\tiny +\hss}}}{dt}}%
{\FRAC{3}{d^{\hbox to 2pt{\rm\tiny +\hss}}}{dt}}%
{\FRAC{3}{d^{\hbox to 2pt{\rm\tiny +\hss}}}{dt}}}}

\def\average#1,#2,{{1\over #2} \sum_{#1}^{#2}}

\def\eye(#1){{\bf(#1)}\quad}

\newtheorem{lemma}{{\bf Lemma}}

\def\eq#1/{(\ref{e:#1})}

\newcommand{\beqn}[1]{\notes{#1}%
\begin{eqnarray} \elabel{#1}}

\newcommand{\eeqn}{\end{eqnarray} }

\newcommand{\beq}[1]{\notes{#1}%
\begin{equation}\elabel{#1}}

\newcommand{\eeq}{\end{equation}}

\def\bdes{\begin{description}}
\def\edes{\end{description}}

\newcounter{rmnum}

\newcounter{anum}

{\end{list}}

\def\ass(#1:#2){(#1\ref{#1:#2})}

\def\ritem#1{
\item[{\sf \ass(\current_model:#1)}]
}

\newenvironment{recall-ass}[1]{%
\begin{description}
\def\current_model{#1}}{
\end{description}
}

\long\def\comment#1{}

\newcommand{\nv}{{\bm n}}

\newcommand{\qv}{{\bm q}}

\newcommand{\wv}{{\bm w}}

\newcommand{\xv}{{\bm x}}
\newcommand{\yv}{{\bm y}}
\newcommand{\zv}{{\bm z}}

\newcommand{\Am}{{\bm A}}
\newcommand{\Bm}{{\bm B}}

\newcommand{\Fm}{{\bm F}}
\newcommand{\Gm}{{\bm G}}
\newcommand{\Hm}{{\bm H}}
\newcommand{\Id}{{\bm I}}

\newcommand{\Km}{{\bm K}}

\newcommand{\Om}{{\bm O}}

\newcommand{\Bc}{{\cal B}}
\newcommand{\Cc}{{\cal C}}

\newcommand{\Jc}{{\cal J}}
\newcommand{\Kc}{{\cal K}}

\newcommand{\Mc}{{\cal M}}
\newcommand{\Nc}{{\cal N}}

\newcommand{\Tc}{{\cal T}}

\newcommand{\Sigmam}{\hbox{\boldmath$\Sigma$}}

\renewcommand{\det}{{\hbox{det}}}

\begin{document}

% paper title
\title{Distributed Information Bottleneck for a Primitive Gaussian Diamond MIMO Channel}

\author{
	\IEEEauthorblockN{Yi Song\IEEEauthorrefmark{1}, 
           Hao Xu\IEEEauthorrefmark{2},
		Kai-Kit Wong\IEEEauthorrefmark{2},
		Giuseppe Caire\IEEEauthorrefmark{1},
		and
		Shlomo Shamai (Shitz)\IEEEauthorrefmark{3}
	}
\IEEEauthorblockA{\IEEEauthorrefmark{1}Faculty of Electrical Engineering and Computer Science, Technical University of Berlin, 10587 Berlin, Germany}
\IEEEauthorblockA{\IEEEauthorrefmark{2}Department of Electronic and Electrical Engineering, University College London, London WC1E7JE, U.K.}
\IEEEauthorblockA{\IEEEauthorrefmark{3}Department of Electrical and Computer Engineering, Technion-Israel Institute of Technology, Haifa 3200003, Israel}
\IEEEauthorblockA{E-mail: yi.song@tu-berlin.de, hao.xu@ucl.ac.uk; kai-kit.wong@ucl.ac.uk; caire@tu-berlin.de; sshlomo@ee.technion.ac.il}
\thanks{The corresponding author is Hao Xu. }
}

% make the title area
\maketitle

\begin{abstract}
This paper considers the distributed information bottleneck (D-IB) problem for a primitive Gaussian diamond channel with two relays and MIMO Rayleigh fading.
The channel state  is an independent and identically distributed (i.i.d.) process known at the relays but unknown to the destination. 
The relays are oblivious, i.e., they are unaware of the codebook and treat the transmitted signal as a random process with known statistics. The bottleneck constraints prevent the relays to communicate the channel state information (CSI) perfectly to the destination. 
To evaluate the bottleneck rate, we provide an upper bound by assuming that the destination node knows the CSI and the relays can cooperate with each other, and also two achievable schemes with simple symbol-by-symbol relay processing and compression. Numerical results show that the lower bounds obtained by the proposed achievable schemes can come close to the upper bound on a wide range of relevant system parameters.
\end{abstract}

% \begin{IEEEkeywords}
% 	Distributed information bottleneck, primitive Gaussian diamond channel, oblivious relay, Rayleigh fading, MIMO.
% \end{IEEEkeywords}

\IEEEpeerreviewmaketitle

%%%%%%%%%%%%%%%%%%%%%%%%%%%%%%%%%%%%%%%%%%%%%%%%
\section{Introduction} \label{introduction}

In modern wireless communication systems, the functionalities of a base station have been distributed 
according to a functional split between {\em Radio Units} (RUs), containing the transmission hardware  
(antennas, amplifiers, up/down  frequency and A/D conversion) and the so-called {\em Decentralized Units} (DUs), 
implementing the physical layer and MAC layer (channel coding/decoding, modulation/rate selection, power allocation, etc.). 
RUs and DUs are connected by a {\em fronthaul} network of large but finite capacity \cite{6897914, 7018201, 8387197, singh2020evolution}. 
In this context, the RUs operate as relays, defining a  multiaccess-relay network (uplink) or broadcast-relay network (downlink). A simplified model consists of a single DU, connected to the RUs by individual non-interfering error-free links of given capacity. Such model is referred to as ``primitive'' relay network \cite{kim2007coding, 6825845,8060585} 
% {\RED [ADD REF ON PRIMITIVE RELAY CHANNEL .. ORIGINAL PAPER I THINK BY WEI YU OR YOUNG-HAN KIM .. I DON'T RECALL]} 
and when
only one user is considered, as a ``diamond'' relay network. In addition, under the RU-DU functional split, 
the RUs (i.e., the relays) are {\em oblivious}, i.e., they are unaware of the codebook used to transmit information, and can only treat the transmitted signal as a random process of given statistics (for an information theoretic definition of 
oblivious relaying please see \cite{4544988,sanderovich2009distributed, aguerri2019TIT, katz2019gaussian, katz2021filtered, caire2018information, IBxu}).
% {\RED [REF TO SME OF SHLOMO'S PAPERS ON OBLIVIOUS RELAYING INCLUDING
% ]}). 

In the case of a single user and single relay, the capacity of such channel coincides with the solution of the 
so-called {\em information bottleneck} (IB) problem introduced by Tishby in \cite{tishby2000information}, 
where we wish to maximize $I(X; Z)$ subject to the bottleneck constraint $I(Y;Z) \leq C$ with $X$, $Y$, $Z$, $C$ respectively being the channel input, observation at the relay, representation variable communicated by the relay to the destination, and the capacity of the relay-to-destination  link. In the case of multiple users and relays, the problem has been generalized in many ways
(e.g., see  \cite{winkelbauer2014rate, winkelbauer2014ratevec, 4544988, sanderovich2009distributed, aguerri2019TIT, katz2019gaussian, katz2021filtered, courtade2013multiterminal, estella2018distributed, aguerri2019distributed, caire2018information, info12040155, IBxu}).
In particular, a general expression for the capacity region of the multi-access multi-relay case 
was found in  \cite{aguerri2019TIT}, under the additional condition that the signals $Y_k$ received at the $k$-th relays are mutually conditionally independent given by the transmitted signals. 

In wireless communication, the knowledge of the channel state (i.e., the matrix of channel coefficients between transmit and receiving antennas) is crucial to enable coherent detection. Accurate channel state information (CSI) can be obtained due to the fact that the channel coefficients (that form a random process that varies in time and frequency) remain practically constant over time-frequency blocks spanning a certain number $N_c$ of time-frequency channel uses.\footnote{The channel coherence block length $N_c$ is approximately given by $\alpha \lceil T_c W_c \rceil$ where $W_c$  (in Hz) is the channel coherence bandwidth, and $T_c$ is the channel coherence time. In turns, $W_c$ depends on the inverse of 
channel delay spread, and $T_c$ depends on the inverse of the channel Doppler spread, and $\alpha$ is some system constant  $\leq 1$. In typical wireless/mobile communications operating outdoor, in the carrier frequency range between 2 and 6 GHz, and with user mobility up to a few tens of km/h, $N_c$ may vary from a few hundred to a few thousands of symbols. For practical reasons, though, actual systems perform channel estimation on much shorter blocks (the so-called resource blocks) of 
$12 \times 14 = 168$ symbols, specified in standards such as 4G-LTE and 5GNR \cite{dahlman20134g, hui2018channel, krasniqi2018performance}. 
% {\RED [ADD SOME REFERENCE]}. 
This can be regarded as a sort of one-size fits all worst case design, also due to the granularity of the multiuser scheduling, that imposes the allocation of rather short data blocks.}
In this paper we are concerned with the uplink. In this case, the user sends some pilot symbols in each coherence block to allow the receiver to estimate the CSI. For a user with $M$ antennas, $M$ mutually orthogonal pilot sequences must be transmitted simultaneously from the antennas, requiring a minimum pilot length (in time-frequency channel uses) of $M$. In large MIMO systems, $M$ may be comparable with $N_c$. Thus, communicating the pilot field (from the relays to the receiver) over the capacity constrained fronthaul links imposes a non-trivial cost in terms of rate. 

The question that we pose in this paper is whether some oblivious (local) processing at the relays can be used in order to alleviate the burden of communicating the quantized pilot field over the fronthaul. 
In order to make the problem more tractable, we consider that the CSI is given for free (genie-aided) at the relays, but not at the
destination. Hence, the relays have the option of compressing the CSI and send it through the fronthaul links together with the received signal, or use the local CSI to operate some processing to the received signal, such that this can be further decoded at the destination without the explicit need of CSI. 
In particular, in this paper we consider an upper bound obtained by letting the CSI be known also at the destination (the so-called ``informed receiver'' upper bound), and compare it with some achievability strategies based on simple oblivious local processing. 
Interestingly, we find that under certain conditions the achievable lower bounds come quite close to the (unachievable) upper bound. This suggests that some ``intelligent'' oblivious processing at the relays may be useful in the RU-DU distributed base station paradigm, rather than insisting on ``dumb antennas'' \cite{viswanath2002opportunistic, 984690}.
% {\RED [REF SOME PAPER THAT TALKS ABOUT DUMB ANTENNAS]}. 

%%%%%%%%%%%%%%%%%%%%%%%%%%%%%%%%%%%%%%%%%%%%%%%%%%
\section{Problem Formulation}
\label{problem_formu}

\begin{figure}
	\centering
	\includegraphics[scale=0.45]{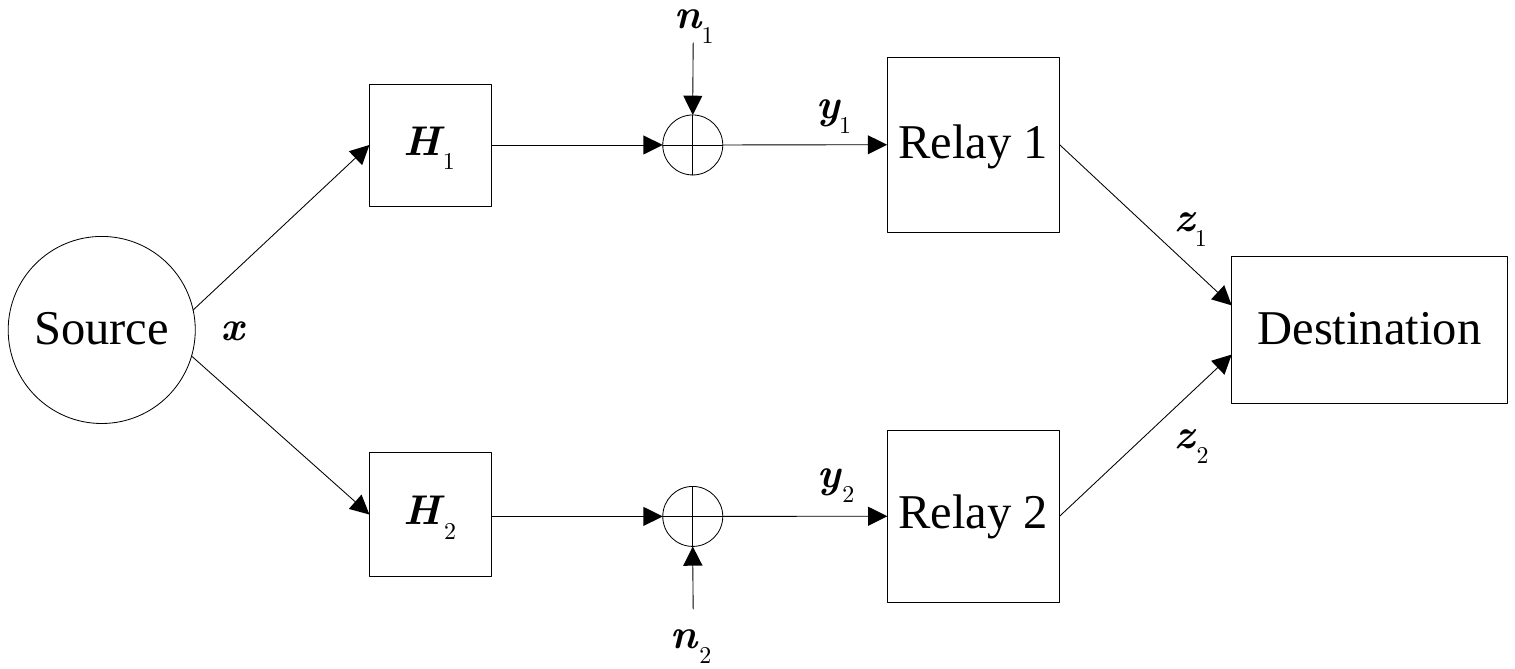}
	\vspace{-1em}
	\caption{A primitive Gaussian diamond MIMO channel with two relays.}
	\label{Block_diagram}
\end{figure}

As shown in Fig.~\ref{Block_diagram}, this paper considers a primitive Gaussian diamond channel with two relays and studies the distributed information bottleneck (D-IB) problem. 
The source node transmits signal $\xv \in \mathbb{C}^{M \times 1}$ to the relays over Gaussian MIMO channel with i.i.d. Rayleigh fading and each relay is connected to the destination via an error-free link with capacity  $C_k, ~\forall~ k \in {\cal K} \triangleq \{ 1, 2 \}$. 
The observation of relay $k$ is
\begin{equation}\label{obser}
\yv_k = \Hm_k \xv + \nv_k,
\end{equation}
where $\xv$ and $\nv_k \in \mathbb{C}^{N_k \times 1}$ are, respectively, zero-mean circularly symmetric complex Gaussian input and noise at relay $k$ with covariance matrix $\Id_M$ and $\sigma^2 \Id_{N_k}$, i.e., $\xv \sim \Cc\Nc(\mathbf{0}, \Id_M)$ and $\nv_k \sim \Cc\Nc(\mathbf{0}, \sigma^2 \Id_{N_k})$. $\Hm_k \in \mathbb{C}^{N_k \times M}$, which denotes channel fading from the source to relay $k$, is a random matrix independent of both $\xv$ and $\nv_k$, and the elements of $\Hm_k$ are i.i.d. 
$\sim \Cc\Nc(0, 1)$ (i.e., zero-mean unit-variance complex circularly symmetric Gaussian).

The relays are constrained to operate without knowledge of the codebooks, i.e., they perform oblivious processing and forward representations of their observations $\zv_k$ to the destination.
According to \cite[Theorem~$1$]{4544988}, with the bottleneck constraints satisfied, the achievable communication rate at which the source node could encode its messages is upper bounded by the mutual information between $\xv$ and $\zv_{\cal K} = \{\zv_k\}_{k \in {\cal K}}$.
Hence, we consider the following D-IB problem
\begin{subequations}\label{IB_problem}
	\begin{align}
	\mathop {\max }\limits_{\{p(\zv_k| \yv_k, \Hm_k)\}} & I(\xv; \zv_{\cal K}) \label{IB_problem_a}\\
	\text{s.t.} \quad\;\;\; &  I(\yv_{\cal T}, \Hm_{\cal T}; \zv_{\cal T}| \zv_{\overline{{\cal T}}}) \leq \sum_{k \in {\cal T}} C_k, ~\forall~ {\cal T} \subseteq {\cal K}, \label{IB_problem_b}
	\end{align}
\end{subequations}
where $C_k$ is the bottleneck constraint of relay $k$ and $\overline{{\cal T}}$ is the complementary set of ${\cal T}$, i.e., $\overline{{\cal T}} = {\cal K} \setminus {\cal T}$.
We call $I(\xv; \zv_{\cal K})$ the bottleneck rate and $I(\yv_{\cal T}, \Hm_{\cal T}; \zv_{\cal T}| \zv_{\overline{{\cal T}}})$ the compression rate.
Since the channel coefficient $\Hm_k$ varies in each realization and is only known at the relay $k$, $\Hm_{\cal T}$ is included in the compression rate formulation.
In (\ref{IB_problem}), we aim to find conditional distributions $p(\zv_k| \yv_k, \Hm_k), \forall k \in {\cal K}$ such that collectively, the compressed signals at the destination preserve as much the original information from the source as possible.

\section{Informed Receiver Upper Bound}
\label{informed_ub}
Since it is hard to derive a closed-form solution to the problem \eqref{IB_problem}, we derive an upper bound in this section.
Similar to the one-relay IB problems studied in \cite{caire2018information, info12040155, IBxu}, an obvious upper bound to problem (\ref{IB_problem}) can be obtained by assuming that the destination node knows all the channel coefficients $\Hm_{\cal K} = \{\Hm_k\}_{k \in {\cal K}}$.
We call this bound the informed receiver upper bound.
The D-IB problem then becomes
\begin{subequations}\label{IB_problem_ub}
	\begin{align}
	\mathop {\max }\limits_{\{p(\zv_k| \yv_k, \Hm_k)\}} & I(\xv; \zv_{\cal K}| \Hm_{\cal K}) \label{IB_problem_ub_a}\\
	\text{s.t.} \quad\;\;\; &  I(\yv_{\cal T}; \zv_{\cal T}| \zv_{\overline{{\cal T}}}, \Hm_{\cal K}) \leq \sum_{k \in {\cal T}} C_k, ~\forall~ {\cal T} \subseteq {\cal K}. \label{IB_problem_ub_b}
	\end{align}
\end{subequations}

Unlike the MIMO channel with one relay in \cite{IBxu}, it is still difficult to solve \eqref{IB_problem_ub}. Hence, besides the assumption that the destination node knows $\Hm_{\cal K}$, we further assume that the relays can cooperate such that each relay also knows the observations $\yv_k$ and $\Hm_k$ of the other relay.
Actually, the network in this case can be seen as a system with a source node with $M$ antennas, a relay with $\sum_{k \in \mathcal{K}} N_k$ antennas, a destination node, and bottleneck constraint $\sum_{k \in \mathcal{K}} C_k$, so the problem (\ref{IB_problem_ub}) becomes
\begin{subequations}\label{IB_problem_ub2}
	\begin{align}
	\mathop {\max }\limits_{\{p(\zv_k| \yv_k, \Hm_k)\}} \quad & I(\xv; \zv_{\cal K}| \Hm_{\cal K}) \label{IB_problem_ub2_a}\\
	\text{s.t.} \quad\quad\;\;\; &  I(\yv_{\cal K}; \zv_{\cal K}| \Hm_{\cal K}) \leq \sum_{k \in \mathcal{K}} C_k. \label{IB_problem_ub2_b}
	\end{align}
\end{subequations}
Denote matrix $\Hm = [\Hm_1; \Hm_2] \in \mathbb{C}^{(\sum_{k \in \mathcal{K}} N_k) \times M}$.
% and $\yv = [\yv_1,\yv_2]^T \in \Cc^{(N_1 + N_2) \times 1}$, then the problem becomes 
% \begin{subequations}\label{IB_problem_ub22}
% 	\begin{align}
% 	\mathop {\max }\limits_{p(\zv| \yv, \Hm)} \quad & I(\xv; \zv| \Hm) \label{IB_problem_ub22_a}\\
% 	\text{s.t.} \quad\quad\;\;\; &  I(\yv; \zv| \Hm) \leq \sum_{k \in \mathcal{K}} C_k. \label{IB_problem_ub22_b}
% 	\end{align}
% \end{subequations}
%
Obviously, the matrix $\Hm \Hm^{H}$ has $T = \min(
\sum_{k \in \Kc} N_k, M)$ positive eigenvalue $\lambda$.
It is known from \cite[(A17)]{IBxu} that the probability density function (pdf) of unordered eigenvalue $\lambda$ 
 of $\Hm \Hm^{H}$ is
\begin{equation}\label{pdf_lambda}
f_{\lambda} (\lambda) = \frac{1}{T} \sum_{i=0}^{T-1} \frac{i!}{(i + S - T)!} [L_{i}^{S-T} (\lambda)]^2 \lambda^{S-T} e^{(-\lambda)},
\end{equation}
where $S = \max{(\sum_{k \in \Kc} N_k, M)}$ and the Laguerre polynomials are 
\begin{align}
L_{i}^{S-T} (\lambda) = \frac{e^{\lambda}}{i! \lambda^{S - T}}  \frac{d^i}{d \lambda^i} \left(e^{-\lambda} \lambda^{S - T + i} \right).
\end{align}
Then, according to \cite[Theorem~1]{IBxu}, the solution of problem (\ref{IB_problem_ub2}), which forms an upper bound to the bottleneck rate $I(\xv; \zv_{\cal K})$ in (\ref{IB_problem_a}), is given by
\begin{equation}\label{R_up_KM}
R^{\text {ub}} = T \int_{\nu \sigma^2}^{\infty} \left[ \log \left(1 + \frac{\lambda}{\sigma^2} \right) - \log (1 + \nu)\right] f_\lambda (\lambda) d \lambda,
\end{equation}
where $\nu$ is chosen such that the following bottleneck constraint is met
\begin{equation}\label{bottle_constr_KM}
\int_{\nu \sigma^2}^{\infty} \left( \log \frac{\lambda}{\nu \sigma^2} \right) f_\lambda (\lambda) d \lambda = \frac{C_1 + C_2}{T}.
\end{equation}

% \section{Scalar Case}

\section{Achievable Schemes}
\label{achiev_schems}

In this section, we provide two achievable schemes where each scheme satisfies the bottleneck constraint and gives a lower bound to the bottleneck rate. 
Before that, we first give a result from \cite[Theorem~$5$]{4544988}, which is important for deriving the achievable schemes. 

Note that in \cite[Theorem~$5$]{4544988}, where the source and each relay only have a single antenna with fixed constant channel information $h_k, \forall k \in {\cal K}$ perfectly known at the destination node, the optimal value of problem (\ref{IB_problem_ub}) is 
\begin{align}\label{R_fixed_rho}
& R (\rho_{\cal K}, C_{\cal K}) = \nonumber\\
&\mathop {\max }\limits_{\{ r_k \}}\! \left\{\! \mathop {\min }\limits_{ {\cal T} \subseteq {\cal K}} \left\{\! \log\! \left[ 1 \!\!+\!\!\! \sum_{k \in {\cal T}^C} \!\rho_k \!\left( 1 \!\!-\!\! 2^{-r_k} \right) \right] \!\!+\!\! \sum_{k \in {\cal T}} (C_k \!-\! r_k) \!\right\} \!\right\}\!,
\end{align}
where $C_{\cal K} = \{C_k\}_{k \in {\cal K}}$, $\rho_{\cal K} = \{\rho_k\}_{k \in {\cal K}}$, $\rho_k = |h_k|^2/\sigma^2$ is the channel signal-to-noise ratio (SNR), and $r_k \geq 0$ is an intermediate variable.
%
% Notice that if we formulate problem (\ref{IB_problem}) based on \cite[Theorem~$1$]{courtade2013multiterminal}, (\ref{R_fixed_rho}) can also be obtained by using \cite[Theorem~$3$]{courtade2013multiterminal} and \cite[Theorem~$2$]{estella2018distributed}.
%
The optimal bottleneck rate $R (\rho_{\cal K}, C_{\cal K})$ in \eqref{R_fixed_rho} can be  obtained by introducing an auxiliary variable $\beta$ to solve the following equivalent problem
\begin{subequations}\label{eq_problem}
	\begin{align}
	& \mathop {\max }\limits_{r_1, r_2, \beta} \; \beta \label{eq_problem_a}\\
	& \text{s.t.} \; \log\! \left[\! 1 \!+\! \sum_{k \in {\cal T}^C} \rho_k \left( 1 \!-\! 2^{-r_k} \right) \!\right] \!+\! \sum_{k \in {\cal T}} (C_k \!-\! r_k) \!\geq\! \beta, \forall~ {\cal T} \!\subseteq\! {\cal K}, \label{eq_problem_b}\\
	& \quad\quad 0 \leq r_k \leq C_k, ~\forall~ k \in {\cal K}. \label{eq_problem_c}
	\end{align}
\end{subequations}
It can be readily found that problem \eqref{eq_problem} is convex and can thus be optimally solved using tools like CVX. 
In the following two subsections, we give the achievable schemes.

\subsection{Quantized channel inversion (QCI) scheme when $M \leq \min_{k \in \mathcal{K}} N_k$}
\label{QCI_scheme}

In our first scheme, each relay first gets an estimate of the channel input using channel inversion and then transmits the quantized noise levels as well as the compressed noisy signal to the destination node.

In particular, we apply the pseudo inverse matrix of $\Hm_k$, i.e., $(\Hm_k^H \Hm_k)^{-1} \Hm_k^H$, to $\yv_k$ and obtain the zero-forcing estimate of $\xv$ for relay $k$ as follows:
\begin{align}\label{eq:pinverse_x}
    \tilde{\xv}_k &=  (\Hm_k^H \Hm_k)^{-1} \Hm_k^H \yv_k \nonumber \\
    &= \xv + (\Hm_k^H \Hm_k)^{-1} \Hm_k^H \nv_k = \xv + \tilde{\nv}_k.
\end{align}
For a given channel matrix $\Hm_k$, $\tilde{\nv}_k \sim \mathcal{CN}(\mathbf{0}, \Am_k)$, where $\Am_k = \sigma^2 (\Hm_k^H \Hm_k)^{-1}$. Let $\Am_k = \Am_k^{(1)} + \Am_k^{(2)}$, where $\Am_k^{(1)} = \Am_k \odot \Id_{K} = \text{diag}\{a_{k, 1},..., a_{k, M}\}$, $a_{k, i}, \forall i \in \Mc$ is the $i$-th diagonal element of $\Am_k$,  and $\Am_k^{(2)} = \Am_k - \Am_k^{(1)}$. Since $\Hm_k$ follows a non-degenerate continuous distribution and the bottleneck constraint is finite, it is impossible to perfectly transmit the channel information to the destination as in \eqref{IB_problem_ub2}. To reduce the number of bits per channel use required for informing the destination node of the channel information, we only convey a compressed version of $\Am_k^{(1)}$.
We fix a finite grid of $J$ positive quantization points ${\cal B} = \{ b_1, \cdots, b_J \}$, where $b_1 \leq b_2 \leq \cdots \leq b_{J-1} < b_J$, $b_J = + \infty$, and define the following ceiling operation
\begin{equation}\label{ceiling}
\ceil[\big]{a}_{\cal B} = \min_{b \in {\cal B}} \{ a \leq b  \}.
\end{equation}
Then, each relay forces the noise power in sub-channels in  (\ref{eq:pinverse_x}) to belong to a finite set of quantized levels by adding artificial noise, i.e., $\tilde{\nv}_k' \sim \mathcal{CN}(\mathbf{0}, \text{diag}\{\ceil[\big]{a_{k, 1}}_{\cal B} - a_{k, 1}, ..., \ceil[\big]{a_{k, M}}_{\cal B} - a_{k, M}\})$, which is independent of $\xv$ and $\tilde{\nv}_k$. Hence the degraded version of $\tilde{\xv}_k$ can be obtained as follows, 
\begin{align}\label{eq:x_hat_k}
    \hat{\xv}_k &= \tilde{\xv}_k + \tilde{\nv}_k' = \xv + \tilde{\nv}_k + \tilde{\nv}_k' = \xv + \hat{\nv}_k, 
\end{align}
where $\hat{\nv}_k \sim \mathcal{CN}(\mathbf{0}, \tilde{\Am}_k^{(1)} + \Am_k^{(2)})$ for a given $\Hm_k$, $\tilde{\Am}_k^{(1)}  \triangleq  \text{diag}\{\ceil[\big]{a_{k, 1}}_{\cal B}, ..., \ceil[\big]{a_{k, M}}_{\cal B} \}$. 

Due to $\Am_k^{(2)}$, the elements in the noise vector $\hat{\nv}_k$ are correlated. 
To evaluate the bottleneck rate, we consider a new auxiliary variable 
\begin{align}\label{eq:X_hat_g}
    \hat{\xv}_k^{\rm g} = \xv + \hat{\nv}_k^{\rm g},
\end{align}
where $\hat{\nv}_k^{\rm g} \sim \mathcal{CN}(\mathbf{0}, \tilde{\Am}_k^{(1)})$. Notice that \eqref{eq:X_hat_g} can be seen as $M$ parallel scalar Gaussian sub-channels with noise power $\ceil[\big]{a_{k, i}}_{\cal B}$ for sub-channel $i$. Since each quantized noise level $\ceil[\big]{a_{k, i}}_{\cal B}$ only has $J$ possible values, it is possible for the relay to inform the destination node of the channel information via the constrained link. 
Then, according to \cite[(129)]{4544988}, 
the optimal representation of $ \hat{\xv}_k^{\rm g}$ given bottleneck constraint $C_k$ is
\begin{align}\label{eq:zv_k_g}
    \hat{\zv}^{\rm g}_{k} = \hat{\xv}_k^{\rm g} + \hat{\wv}_k,
\end{align}
where $\hat{\wv}_k$ is the complex Gaussian distribution with mean $\mathbf{0}$ and a diagonal covariance matrix $\Sigmam_{\hat{\wv}_k}$ whose elements are determined by the SNR of each sub-channel, i.e., the diagonal elements of $\tilde{\Am}_k^{(1)}$. 
We also add the noise vector $\hat{\wv}_k$ to $\hat{\xv}_k$ in (\ref{eq:x_hat_k}) and obtain its representation as follows
\begin{align}\label{eq:zv_k}
    \hat{\zv}_k = \hat{\xv}_k + \hat{\wv}_k.
\end{align}
Then, we have the following lemma.
\begin{lemma}\label{lem:LB_QCI}
    If $\tilde{\Am}_k^{(1)}$ is forwarded to the destination node for each channel realization by relay $k, ~\forall~ k \in \Kc$, with signal vectors $\hat{\xv}_k$ and $\hat{\xv}_k^{\rm g}$ in \eqref{eq:x_hat_k} and \eqref{eq:X_hat_g}, and their representations denoted as $\hat{\zv}_k$ in \eqref{eq:zv_k} and $\hat{\zv}_k^{\rm g}$ in \eqref{eq:zv_k_g},$ ~\forall~ \Tc \subseteq  \Kc$,  we have  
    \begin{align}
    I(\hat{\xv}_{\cal T}; \hat{\zv}_{\cal T}| \hat{\zv}_{\overline{{\cal T}}}, \tilde{\Am}_{\cal K}^{(1)}) &\leq I(\hat{\xv}_{\cal T}^{\rm g}; \hat{\zv}_{\cal T}^{\rm g}| \hat{\zv}_{\overline{{\cal T}}}^{\rm g}, \tilde{\Am}_{\cal K}^{(1)}), \label{eq:a} \\
    I(\xv; \hat{\zv}_{\Kc} | \tilde{\Am}_{\cal K}^{(1)}) &\geq I(\xv; \hat{\zv}_{\Kc}^{\rm g} | \tilde{\Am}_{\cal K}^{(1)}). \label{eq:b}
\end{align}
\end{lemma}
\begin{proof}
       Due to space limitation, the proof is provided in Appendix A in \cite{}, which is a long version of this paper.
\end{proof}

Based on Lemma \ref{lem:LB_QCI}, a lower bound to the D-IB problem \eqref{IB_problem} can be obtained by solving the following problem
\begin{subequations}\label{IB_problem_QCI}
	\begin{align}
	\mathop {\max }\limits_{\{p(\hat{\zv}^{\rm g}_k| \hat{\xv}^{\rm g}_k, \tilde{\Am}_k^{(1)})\}} & I(\xv; \hat{\zv}^{\rm g}_{\cal K}|\tilde{\Am}_{\cal K}^{(1)}) \label{IB_problem_QCI_a}\\
	\text{s.t.} \quad\;\;\; &  I(\hat{\xv}^{\rm g}_{\cal T}; \hat{\zv}^{\rm g}_{\cal T}| \hat{\zv}^{\rm g}_{\overline{{\cal T}}}, \tilde{\Am}_{\cal K}^{(1)}) \leq \sum_{k \in {\cal T}} (C_k - B_k), \nonumber \\
    &~~~~ ~\forall~ {\cal T} \subseteq {\cal K}, \label{IB_problem_QCI_b}
	\end{align}
\end{subequations}
where $B_k$ is the number of bits required for compressing $\tilde{\Am}_{k}^{(1)}$.

We define a space $\Xi = \{(j_1, ..., j_M), ~\forall j_i \in \{1, 2, ..., J\} , i \in \{1, 2, ..., M\}\}$. In total there are $J^M$ points in the space. Let $\xi = (j_1, ..., j_M)$ denote a point in the space $\Xi$.
Its probability mass function is given by 
\begin{align}
    P_{\xi} = \text{Pr}\{\ceil[\big]{a_{k, 1}}_{\cal B} = b_{j_1}, ..., \ceil[\big]{a_{k, M}}_{\cal B} = b_{j_M}\}.
\end{align}
The joint entropy of  $\ceil[\big]{a_{k, i}}_{\cal B}, ~\forall~ i \in \mathcal{M}$ , i.e., the minimum number of bits to jointly source-encode $\ceil[\big]{a_{k, i}}_{\cal B}, ~\forall~ i \in \mathcal{M}$, is thus given by 
\begin{align}\label{eq:H_joint}
    H_{\text{joint}}^k = \sum_{\xi \in \Xi} -P_{\xi} \log(P_{\xi}).
\end{align}
However, it is difficult to obtain the joint entropy $H_{\rm joint}^k$ from \eqref{eq:H_joint}, as there are $J^{M}$ points in space $\Xi$. To reduce the complexity, we consider the (slightly) suboptimal but far more practical entropy coding of each noise level $\ceil[\big]{a_{k, i}}_{\cal B}, ~\forall~ i \in \Mc$ separately and obtain the upper bound of $H_{\rm joint}^k$ as 
\begin{align}
    H_{\text{sum}}^k &= \sum_{i=1}^M H_i^k = M H^k \nonumber \\
    & = \small -M \sum_{j_k = 1}^J \text{Pr} \{\ceil[\big]{a_{k}}_{\cal B} \!=\! b_{j_k}\} \log(\text{Pr} \{(\ceil[\big]{a_{k}}_{\cal B} \!=\! b_{j_k}) \},\!\!\!
\end{align}
where $H_i^k$ denotes the entropy of  $\ceil[\big]{a_{k, i}}_{\cal B}$ and the second equality holds  since it is stated in \cite[Appendix F]{IBxu} that when $M \leq \min_{k \in \Kc} N_k$, 
%
% $\Hm_k^H \Hm_k \sim \mathcal{CW}_M(\mathbf{0}, \Id_{M})$, matrix $(\Hm_k^H \Hm_k )^{-1}$
%
the matrix $\sigma^2(\Hm_k^H \Hm_k )^{-1}$ follows a complex inverse Wishart distribution and its diagonal elements are identically inverse chi square distributed whose pdf is presented in \cite[(A44)]{IBxu}. 
Hence, $H_1^k = H_2^k = ...= H_M^k$ and we neglect the subscript $i$.  
Therefore, $ \text{Pr} \{\ceil[\big]{a_{k}}_{\cal B} =b_{j_k}\}, ~\forall~ b_{j_k} \in \Bc$ can be computed and the D-IB problem becomes 
\begin{subequations}\label{IB_problem_QCI_3}
	\begin{align}
	\mathop {\max }\limits_{\{p(\hat{\zv}^{\rm g}_k| \hat{\xv}^{\rm g}_k, \tilde{\Am}_k^{(1)})\}} & I(\xv; \hat{\zv}^{\rm g}_{\cal K}|\tilde{\Am}_{\cal K}^{(1)}) \label{IB_problem_QCI_aaa}\\
	\text{s.t.} \quad\;\;\; &  I(\hat{\xv}^{\rm g}_{\cal T}; \hat{\zv}^{\rm g}_{\cal T}| \hat{\zv}^{\rm g}_{\overline{{\cal T}}}, \tilde{\Am}_{\cal K}^{(1)}) \!\leq \!\sum_{k \in {\cal T}} (C_k - H_{\rm sum}^k), \nonumber \\
	&~~~~~~~~\forall~ {\cal T}\subseteq {\cal K}. \label{IB_problem_QCI_bbb}
	\end{align}
\end{subequations}

Based on the definition of $\hat{\xv}_k^{\rm g}$ in \eqref{eq:X_hat_g}, the relay $k$ connects the source $\xv$ through $M$ independent parallel Gaussian sub-channels, where noise power follows the same distribution. Therefore, for each sub-channel of relay $k$, the capacity constraint to the destination can be denoted as $\frac{C_k - H_{\rm sum}^k}{M}$. The problem thus can be simplified as one scalar source and two relays with a single antenna problem as stated in \cite{xudistributed}. 

We denote the quantized SNR of certain sub-channel for relay $k$ in \eqref{eq:X_hat_g} when $\ceil[\big]{a_k}_{\cal B} = b_{j_k}$ by
\begin{equation}\label{rho_hat}
{\hat \rho}_{k, j_k} = \frac{1}{b_{j_k} \sigma^2}, ~\forall~ k \in \Kc, j_k \in {\cal J},
\end{equation}
where ${\cal J} = \{ 1, \cdots, J \}$, and define probability
\begin{equation}\label{P_jk}
{\hat P}_{k, j_k} = {\text {Pr}} \left\{ \ceil[\big]{a_k}_{\cal B} = b_{j_k} \right\}, ~\forall~ j_k \in {\cal J}.
\end{equation}
Note that from \eqref{rho_hat} and the definition of quantization points in $\Bc$, it is known that  if ${j_k} = J$, ${\hat \rho}_{k, j_k} =0$. In this case, we set $c_{k, j_k} = 0$. Besides, according to \cite[(17) - (19)]{xudistributed}, for $j_1 \in \Jc, j_2 \in \Jc$, the achievable rate $R_{j_1, j_2}$ can be obtained as follows by using (\ref{R_fixed_rho}),
\begin{align}\label{R_lb_QCI_gene}
R_{j_1, j_2} & = \mathop {\max }\limits_{\{ r_{k, j_1, j_2} \}}\! \left\{ \mathop {\min }\limits_{ {\cal T} \subseteq {\cal K}} \!\left\{\! \log\! \left[\! 1 \!+\! \sum_{k \in {\cal T}^C} {\hat \rho}_{k, j_k} \left( 1 \!-\! 2^{-r_{k, j_1, j_2}} \right) \!\right] \right.\right. \nonumber\\
& \left.\left. + \sum_{k \in {\cal T}} \left( c_{k, j_k} - r_{k, j_1, j_2} \right) \right\} \right\}, ~\forall~ j_1,~ j_2 \in {\cal J}.
\end{align}

Therefore, based on \cite{xudistributed}, a lower bound to the bottleneck rate, which we will denote by $R^{\text {lb}1}$, can be obtained by solving the following problem
\begin{subequations}\label{problem_QCI}
	\begin{align}
	\!\!\!\mathop {\max }\limits_{\left\{ c_{k, j_k} \right\}}  & \sum_{j_1=1}^J \sum_{j_2=1}^J M {\hat P}_{1, j_1} {\hat P}_{2, j_2} R_{j_1, j_2} \label{problem_QCI_a}\\
	\text{s.t.} \quad &  \sum_{j_k=1}^{J-1} {\hat P}_{k, j_k} c_{k, j_k} \leq \frac{C_k - H_{\rm sum}^k}{M}, ~\forall~ k \in {\cal K}, \label{problem_QCI_b}\\
	& c_{k, j_k} \geq 0, ~\forall~ k \in {\cal K},~ j_k \in {\cal J} \setminus J, \label{problem_QCI_c}\\
	& c_{k, J} = 0, ~\forall~ k \in {\cal K}, \label{problem_QCI_d}
	\end{align}
\end{subequations}
which can be solved similarly as (\ref{eq_problem}) by introducing $\beta_{j_1, j_2}$ for each $R_{j_1, j_2}, \forall j_1, j_2 \in {\cal J}$ and thus can be reformulated as a convex optimization problem, which can be solved by  standard convex optimization tools.

\subsection{MMSE-based scheme}
\label{MMSE_scheme}
In this subsection, we assume that each relay $k$ first produces the MMSE estimate of $\xv$ based on $(\yv_k, \Hm_k)$, and then source-encodes this estimate.
In particular, given $(\yv_k, \Hm_k)$, Denote 
\begin{align}
    \Fm_k =  (\Hm_k \Hm_k^H + \sigma^2 \Id_{N_k})^{-1} \Hm_k .
\end{align}
The MMSE estimate of $\xv$ obtained by relay $k$ is
\begin{align}\label{x_bar_k}
    \bar{\xv}_k &= \Fm^{H}_k \yv_k = \Fm^H_k \Hm_k \xv + \Fm_k^H \nv_k.
\end{align}
Taking ${\bar \xv}_k$ as a new observation, we assume that relay $k$ quantizes ${\bar \xv}_k$ by choosing $P_{\zv_k| {\bar \xv}_k}$ to be a conditional Gaussian distribution, i.e., 
\begin{equation}\label{z_bar_k}
\zv_k = {\bar \xv}_k + \qv_k, 
\end{equation}
where $\qv_k \sim \Cc\Nc(\mathbf{0}, D_k \Id_{M})$ and is independent of $\bar{\xv}_k$. 
Now we evaluate the compression rate and also the bottleneck rate.

To evaluate the compression rate and make sure that it satisfies the bottleneck constraint, we introduce an auxiliary Gaussian vector $\bar{\xv}_k^{\rm g}$ with the same second moment as $\bar{\xv}_k$, i.e.,
\begin{align}\label{eq:x_bar_g}
\bar{\xv}_k^{\rm g} &\sim \Cc\Nc(\bm 0, \bm \Sigma_{\bar{\xv}_k^{\rm g}}), \\
\bm \Sigma_{\bar{\xv}_k^{\rm g}} &= \mathbb{E}[\bar{\xv}_k \bar{\xv}_k^{H}] \nonumber \\
 &= \mathbb{E}\left[\Fm_k^H \Hm_k \Hm_k^H \Fm_k + \sigma^2 \Fm_k^H \Fm_k \right]. \label{eq:sigma_x}
\end{align}
Letting $\lambda_k$ denote the unordered eigenvalue of $\Hm_k \Hm_k^H$ whose rank is $T_k = \min(N_k, M)$, then based on the derivation in \cite[(A80), (A82) and (A83)]{IBxu}, \eqref{eq:sigma_x} can be computed as 
\begin{align}\label{eq:sigma_x2}
    \mathbb{E}\left[\Fm_k^H \Hm_k \Hm_k^H \Fm_k + \sigma^2 \Fm_k^H \Fm_k \right]  = \frac{T_k}{M}\mathbb{E}\left[ \frac{\lambda_k}{\lambda_k + \sigma^2}\right] \Id_M,
\end{align}
where $\mathbb{E}\left[ \frac{\lambda_k}{\lambda_k + \sigma^2}\right]$ can be computed based on the pdf of $\lambda_k$, which takes on similar form as \eqref{pdf_lambda}. 
As in (\ref{z_bar_k}), we also quantize $\bar{\xv}_k^{\rm g}$ by adding $\qv_k$ and obtain its representation $\bar{\zv}_k^{\rm g}$ as follows
\begin{align}
    \bar{\zv}_k^{\rm g} = \bar{\xv}_k^{\rm g} + \qv_k. 
\end{align}
Since Gaussian input maximizes the mutual information of a Gaussian additive noise channel,  we have
\begin{equation}\label{mutual_x_bar_z}
I({\bar \xv}_k; \zv_k) \leq I({\bar \xv}_{k}^{\rm g}; \bar{\zv}_{k}^{\rm g}) = \log \det\left( \Id_M + \frac{{\mathbb E} [\bar{\xv}_k \bar{\xv}_k^H ]}{D_k} \right).
\end{equation}
Let 
\begin{equation} \label{LB_C_K}
\log \det\left( \Id_M + \frac{{\mathbb E} [\bar{\xv}_k \bar{\xv}_k^H ]}{D_k} \right) = C_k.
\end{equation}
We thus have
\begin{equation}\label{mutual_XkbarZk_Ck}
I({\bar \xv}_k; \zv_k) \leq C_k.
\end{equation}
Based on \eqref{eq:sigma_x} and \eqref{eq:sigma_x2}, $D_k$ can be calculated as
\begin{align} \label{D_k}
D_k = \frac{\frac{T_k}{M} \mathbb{E}[\frac{\lambda_k}{\lambda_k + \sigma^2}]}{2^{\frac{C_k}{M}} - 1}.
\end{align}

\begin{lemma}\label{lem:LB_MMSE}
    If $I({\bar \xv}_k; \zv_k) \leq C_k, ~\forall~ k \in \Kc$ are satisfied, then the bottleneck constraint of the considered system, i.e., 
\begin{equation}\label{bt_constr}
I({\bar \xv}_{\cal T}; \zv_{\cal T}| \zv_{\overline{{\cal T}}}) \leq \sum_{k \in {\cal T}} C_k, ~\forall~ {\cal T} \subseteq {\cal K},
\end{equation}
can be guaranteed. 
\end{lemma}
\begin{proof}
     Proof is similar to that in \cite{xudistributed}, so it is neglected here due to space limitation. 
\end{proof}

According to Lemma \ref{lem:LB_MMSE}, the bottleneck constraint is satisfied with the proper design of $D_k$ in \eqref{D_k}.
The next step is to evaluate $I(\xv; \zv_1, \zv_2)$. 
We first derive a lower bound to $I(\xv; \zv_1, \zv_2)$ as
\begin{subequations}
\begin{align}
I(\xv; \zv_1, \zv_2)
&= h(\zv_1, \zv_2) - h(\zv_1, \zv_2| \xv) \\
&\geq h(\zv_1, \zv_2| \Hm_1, \Hm_2) - h(\zv_1, \zv_2| \xv) \label{eq:aa} \\
&= h(\zv_1, \zv_2| \Hm_1, \Hm_2) \!-\! h(\zv_1| \xv) \!-\! h(\zv_2| \xv),\label{mutual_XZ1Z2}
\end{align}
\end{subequations}
where \eqref{eq:aa} is satisfied since conditioning reduces differential entropy, and \eqref{mutual_XZ1Z2} holds since $\zv_2$ is independent of $\zv_1$ given $\xv$.
Then, we evaluate the terms in (\ref{mutual_XZ1Z2}) separately.
Since $\xv$, $\nv_k$, and $\qv_k$ are independent variables, $\zv_1$ and $\zv_2$ are jointly Gaussian distributions under the condition of $(\Hm_1, \Hm_2)$. Therefore, according to \eqref{x_bar_k},  \eqref{z_bar_k} and \cite[A80, A82]{IBxu}, the covariance matrix terms $\Km_{i, j}, ~\forall~, i, j \in \Kc$ are given by 
\begin{align}
    \Km_{k, k} &= \mathbb{E}_{\{\zv_k
    \}}\left[\zv_k \zv_k^H | \Hm_1, \Hm_2\right]\nonumber \\
    &\!=\! \Fm_k^H \Hm_k \Hm_k^H \Fm_k \!+\! \sigma^2 \Fm_k^H \Fm_k \!+\! D_k \Id_M , 
    \forall k \in \Kc, \\
    \Km_{i, j} &= \mathbb{E}_{\{\zv_i, \zv_j\}}\left[\zv_i \zv_j^H|\Hm_1, \Hm_2 \right] \nonumber \\
    &\!=\! \Fm_i^H \Hm_i \Hm_j^H \Fm_j ,~\forall~i,j \in \Kc,~ i \neq j, 
\end{align}
where $D_k$ is given in \eqref{D_k}.
Hence, 
\begin{align}\label{mutual_XZ1Z2_S1S2}
& h(\zv_1, \zv_2| \Hm_1, \Hm_2)\nonumber\\
= & {\mathbb E}_{\{\Hm_1, \Hm_2\}} \left[ \log \left( (\pi e)^{2M} \det \left( \begin{bmatrix}  \Km_{1, 1}& \Km_{1, 2}\\\Km_{2, 1} & \Km_{2, 2}\end{bmatrix}
\right) \right) \right].
\end{align}
Moreover, based on the fact that conditioning reduces differential entropy and according to \cite[(81)]{IBxu} and \eqref{z_bar_k}, we have
\begin{align}\label{zk_X}
h(\zv_k| \xv) &= h((\zv_k - \mathbb{E}_{\zv_k}[\zv_k |\xv])| \xv) \nonumber \\
&= h\left((\Fm_k^H \Hm_k - \mathbb{E}[\Fm_k^H \Hm_k]) \xv + \Fm_k^H \nv_k + \qv_k | \xv\right) \nonumber \\
&\leq h\left((\Fm_k^H \Hm_k \!-\! \mathbb{E}[\Fm_k^H \Hm_k]) \xv \!+\! \Fm_k^H \nv_k \!+\! \qv_k \right), 
\end{align}
where $\mathbb{E}_{\zv_k}[\zv_k |\xv] = \mathbb{E}_{\{\Hm_k, \nv_k, \qv_k\}}[\Fm_k^H \Hm_k \xv + \Fm_k^H \nv_k + \qv_k | \xv] = (\mathbb{E}[\Fm_k^H \Hm_k]) \xv$.
Moreover, since the 
Gaussian distribution maximizes the entropy over all distributions with the same variance  \cite{el2011network}, \eqref{zk_X} is upper bound by
\begin{align}\label{zk_X_2}
    h(\zv_k|\xv) \leq \log((\pi e)^M \det(\Gm_k)),
\end{align}
% \begin{align}
%     h\left((\Fm_k^H \Hm_k - \mathbb{E}[\Fm_k^H \Hm_k]) \xv + \Fm_k^H \nv_k + \qv_k \right) \leq ,
% \end{align}
where based on \cite[A84]{IBxu}, $\Gm_k$ is given by
\begin{align}\label{G_mat}
    \Gm_k \!&=\! \mathbb{E}[\Fm_k^H \!\Hm_k \Hm_k^H \!\Fm_k] \!-\!  \left(\mathbb{E}[\Fm_k^H \Hm_k ]\right)^2 \!+\! \sigma^2 \mathbb{E}[\Fm_k^H \Fm_k] \!+\! D_k \Id_M \nonumber \\
    \!&=\! \left\{\frac{T_k}{M} \mathbb{E}\left[\!\frac{\lambda_k}{\lambda_k \!+\! \sigma^2}\!\right] \!-\! \frac{T_k^2}{M^2} \mathbb{E}^2\left[\!\frac{\lambda_k}{\lambda_k \!+\! \sigma^2}\!\right]\! \!+\! D_k \right\} \Id_M.
\end{align}
Substituting (\ref{mutual_XZ1Z2_S1S2}) and (\ref{zk_X_2}) into (\ref{mutual_XZ1Z2}), a lower bound to $I(\xv; \zv_1, \zv_2)$ can be obtained as follows
\begin{align}\label{R_lb_MMSE}
R^{\text {lb}2} & = {\mathbb E}_{\{\Hm_1, \Hm_2\}} \left[ \log  \det \left( \begin{bmatrix}  \Km_{1, 1}& \Km_{1, 2}\\\Km_{2, 1} & \Km_{2, 2}\end{bmatrix}
\right) \right] \nonumber\\
& ~~~~~~~~- \sum_{k=1}^2 \log\left(
\det (\Gm_k)\right).
\end{align}

\section{Numerical Results}
\label{simulation}

In this section, we investigate the two lower bounds obtained by the proposed achievable schemes and compare them with the informed receiver upper bound. In the simulation, we consider $M = N_1 = N_2 = 3$ and assume the same bottleneck constraint, i.e., $C_1 = C_2 = C$. For QCI scheme, we choose the number of quantization points $J = 2^B$, where $B$ denotes the quantization bits, and we choose the quantization levels as quantiles such that we obtain the uniform pmf ${\hat P}_{k, j_k} = \frac{1}{J}, \forall k \in {\cal K}, j_k \in {\cal J}$.
% When performing the TCI scheme, we vary $S_{\text {th}}$ from $0$ to $2$ in step $0.1$ and choose the one which gives the largest rate value.

\begin{figure}
	\centering
	\includegraphics[scale=0.36]{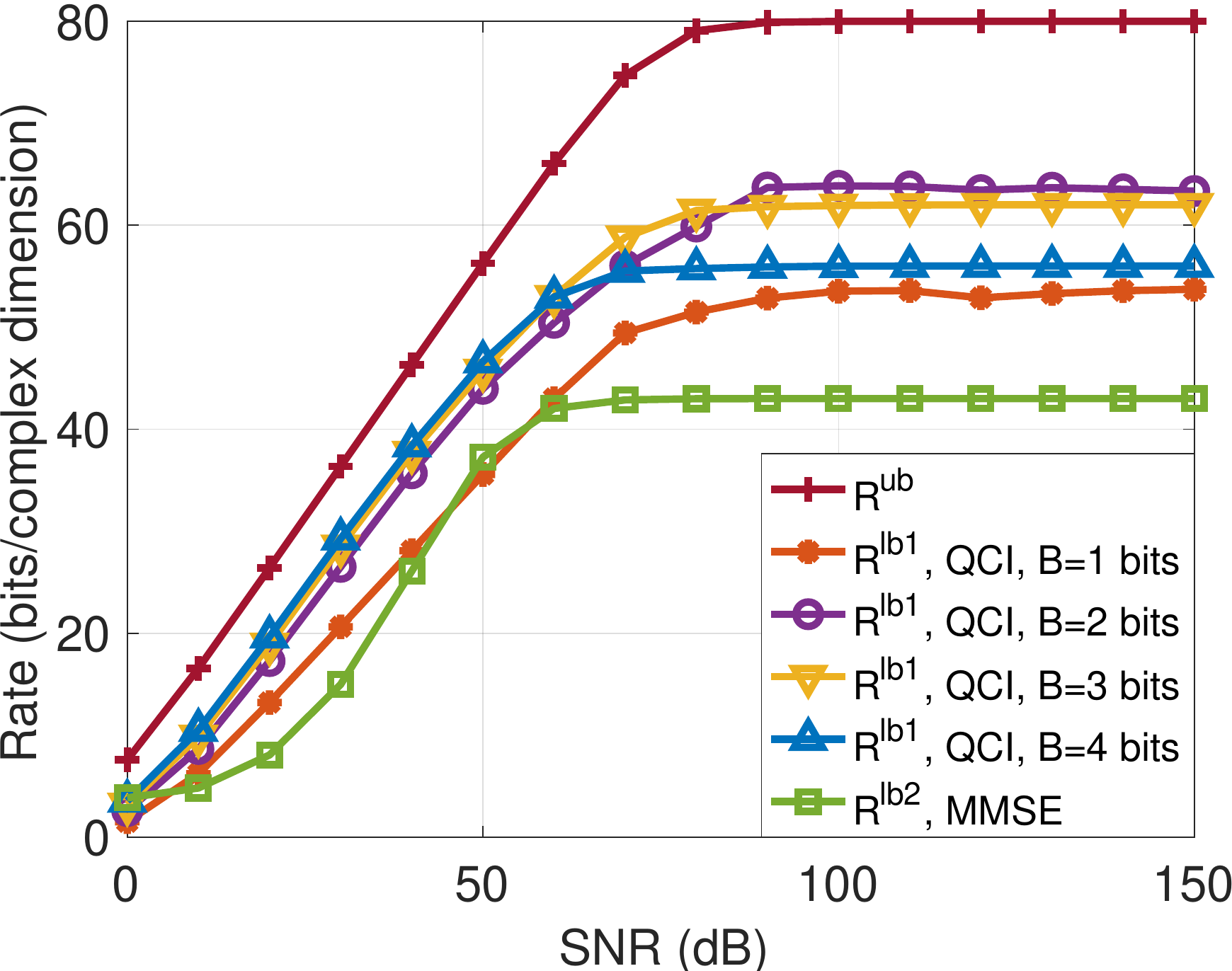}
	\vspace{-1em}
	\caption{Upper and lower bounds to the bottleneck rate versus $\rho$ with $C = 40$ bits/complex dimension.}
	\label{R_VS_rho}
\end{figure}

\begin{figure}
	\centering
	\includegraphics[scale=0.36]{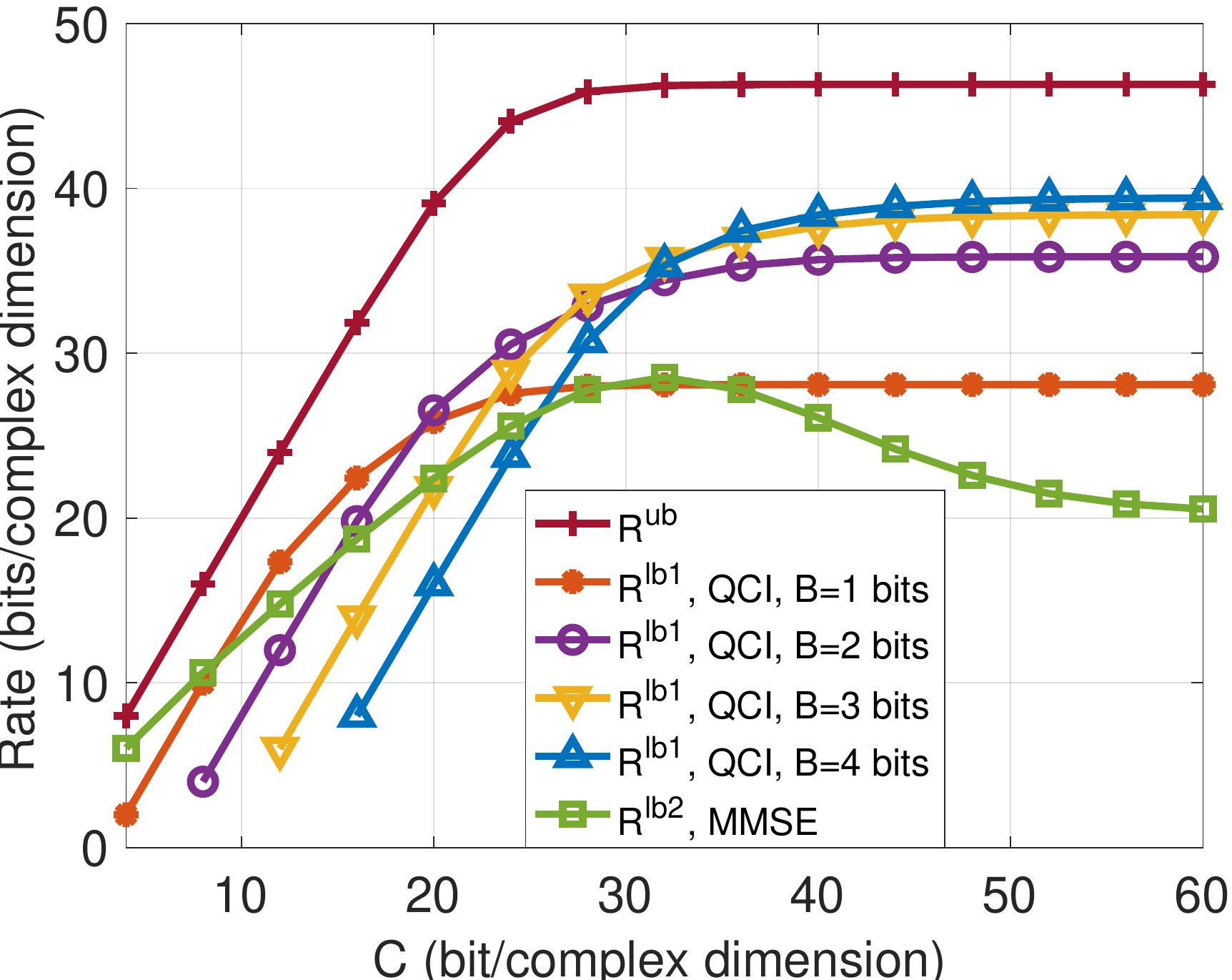}
	\vspace{-1em}
	\caption{Upper and lower bounds to the bottleneck rate versus $C$ with $\rho = 40$~dB.}
	\label{R_VS_C}
\end{figure}

In Fig.~\ref{R_VS_rho}, the upper and lower bounds are plotted against the SNR $\rho = 10 \log_{10}( \frac{1}{\sigma^2}) ~dB$ with $C = 40$ bits/complex dimension. Note that as $\rho$ grows, the upper bound can approach the rate limit, the sum capacity of the two relay-destination links, i.e., $C_1 + C_2$. In addition, for relatively small $\rho$, $R^{\text {lb}1}$ obtained by the QCI scheme with $4$ quantization bits gets quite close to the upper bound, while for relatively large $\rho$, with proper tuning of the quantization bits, the QCI scheme gets closer to the upper bound compared to $R^{\text {lb}2}$.  

Fig.~\ref{R_VS_C} shows the effect of constraint $C$ with SNR $= 40~dB$. 
Since the QCI scheme uses part of the link capacity to transmit the quantized noise power, large quantization bits are impossible for small link capacities. Therefore, QCI schemes with different quantization bits start from different link capacity points.  
As $C$ increases, except for $R^{\text {lb}2}$, the QCI bounds grow monotonically and converge to constants. The QCI has better performance since $R^{\text {lb}1}$ with $4$ quantization bits approximately matches $R^{\text {ub}}$ in the high capacity region. However, $R^{\text {lb}2}$ first increases and then decreases with $C$. From \eqref{D_k}, $D_k$ is monotonically decreasing as $C_k$ increases. Note that in \eqref{R_lb_MMSE} the second term is decreasing because $D_k$ is decreasing, the same as the first term. Therefore, the rate $R^{\text{lb}2}$ is not monotonically increasing as $C$ increases due to two relaxation in \eqref{eq:aa} and \eqref{zk_X}.

\section{Conclusions}
\label{conclusion}

This work extends the IB problem of the two-relay scalar case in \cite{xudistributed} to a Gaussian diamond model with Rayleigh fading MIMO channel.
Due to the bottleneck constraint, the destination node cannot get the perfect CSI from the relays. To evaluate the bottleneck rate, two achievable schemes and an upper bound are derived and compared in simulation in terms of link capacity and SNR.  
Our results show that with simple symbol-by-symbol relay processing and compression, we can obtain a bottleneck rate close to the upper bound for a wide range of relevant system parameters.
In the future, instead of considering the case where the relay has more antennas than the source, we will consider a difficult case where the source has two antennas while the two relays each have one antenna.
\section*{Acknowledgments}
The work of H. Xu has been supported by the European Union's Horizon 2020 Research and Innovation Programme under Marie Skłodowska-Curie Grant No. 101024636 and the Alexander von Humboldt Foundation. And the work of S. Shamai has been supported by the
European Union's Horizon 2020 Research and Innovation Programme,
grant agreement No. 694630 and by the German Research Foundation (DFG) via
the German-Israeli Project Cooperation (DIP), under Project SH 1937/1-1.
\clearpage

\bibliographystyle{IEEEtran}
\bibliography{IEEEabrv,Ref}

\clearpage

\appendices
\section{Proof of Lemma \ref{lem:LB_QCI}}
\label{sec:prove_LB_QCI}

If $\tilde{\Am}_{\cal K}^{(1)}$ is forwarded to destination node, we can derive the upper bound to $I(\hat{\xv}_{\cal T}; \hat{\zv}_{\cal T}| \hat{\zv}_{\overline{{\cal T}}}, \tilde{\Am}_{\cal K}^{(1)}), ~\forall~ \Tc \subseteq  \Kc$ as 
% $I(\hat{\xv}^{\rm g}_{\cal T}; \hat{\zv}^{\text g}_{\cal T}| \hat{\zv}^{\text g}_{\overline{{\cal T}}}, \tilde{\Am}_{\cal K}^{(1)})$
\begin{align}\label{eq:constraint_inequality}
    &I(\hat{\xv}_{\cal T}; \hat{\zv}_{\cal T}| \hat{\zv}_{\overline{{\cal T}}}, \tilde{\Am}_{\cal K}^{(1)}) \nonumber \\
    =& h(\hat{\zv}_{\cal T}| \hat{\zv}_{\overline{{\cal T}}}, \tilde{\Am}_{\cal K}^{(1)}) - h(\hat{\zv}_{\cal T}| 
    \hat{\zv}_{\overline{{\cal T}}}, \hat{\xv}_{\cal T}, \tilde{\Am}_{\cal K}^{(1)}) \nonumber \\
    =& h(\hat{\zv}_{\cal T}| \hat{\zv}_{\overline{{\cal T}}}, \tilde{\Am}_{\cal K}^{(1)}) - h(\hat{\zv}_{\cal T}| \hat{\xv}_{\cal T}, \tilde{\Am}_{\cal K}^{(1)}) \nonumber \\
    \overset{(a)}{=} &h(\hat{\zv}_{\cal T}| \hat{\zv}_{\overline{{\cal T}}}, \tilde{\Am}_{\cal K}^{(1)}) - \sum_{k \in T} h(\hat{\zv}_{k}| \hat{\xv}_{k} \tilde{\Am}_{\cal K}^{(1)}) \nonumber \\
    \overset{(b)}{=} &h(\hat{\zv}_{\cal T}| \hat{\zv}_{\overline{{\cal T}}}, \tilde{\Am}_{\cal K}^{(1)}) - \sum_{k \in T} h(\hat{\zv}_{k}^{\rm g}| \hat{\xv}_{k}^{\rm g}, \tilde{\Am}_{\cal K}^{(1)}),
\end{align}
where $(a)$ holds since $\hat{\zv}_1$ and $\hat{\zv}_2$ are independent given $\hat{\xv}_{\Kc}$ and $(b)$ is satisfied since  $h(\hat{\zv}_k |\hat{\xv}_k, \tilde{\Am}_{\cal K}^{(1)})$ is equal to $\log((\pi e)^M\det(\bm \Sigma_{\hat{\wv}_{k}}))$, which is the same as $h(\hat{\zv}_k^{\rm g} |\hat{\xv}_k^{\rm g}, \tilde{\Am}_{\cal K}^{(1)})$.
In order to prove \eqref{eq:a}, we need to prove $h(\hat{\zv}_{\cal T}| \hat{\zv}_{\overline{{\cal T}}}, \tilde{\Am}_{\cal K}^{(1)}) \leq h(\hat{\zv}_{\cal T}^{\rm g}| \hat{\zv}_{\overline{{\cal T}}}^{\rm g}, \tilde{\Am}_{\cal K}^{(1)}), ~\forall~ \Tc \subseteq \Kc $.
Therefore, we discuss different cases of $\Tc$. First, when $\Tc = \emptyset$, it is obvious that 
\begin{align}\label{T_0}
    I(\hat{\xv}_{\cal T}; \hat{\zv}_{\cal T}| \hat{\zv}_{\overline{{\cal T}}}, \tilde{\Am}_{\cal K}^{(1)})  = I(\hat{\xv}_{\cal T}^{\rm g}; \hat{\zv}_{\cal T}^{\rm g}| \hat{\zv}^{\rm g}_{\overline{{\cal T}}}, \tilde{\Am}_{\cal K}^{(1)}) = 0.
\end{align}
In the second case when $\Tc = \{1, 2\}$, since Gaussian distribution maximizes the differential entropy with the same variance, an upper bound to $h(\hat{\zv}_{\Kc}| \tilde{\Am}_{\cal K}^{(1)})$ is given by
\begin{align}\label{eq:h_zk}
    h(\hat{\zv}_{\Kc}| \tilde{\Am}_{\cal K}^{(1)}) \leq \log((\pi e)^{2 M}  \Sigmam_{\hat{\zv}_{\Kc}}).
\end{align}
Note that $\hat{\zv}_1$ and $\hat{\zv}_2$ are joint Gaussian distribution if the second order of noises, $\mathbf{E}[\hat{\nv}_k \hat{\nv}_k^H], ~\forall~ k \in \Kc$ are known. Assuming that the noises are known, the covariance matrix $ \Sigmam_{\hat{\zv}_{\Kc}}$ in \eqref{eq:h_zk} is given by
\begin{align}\label{eq:sigma_zk}
   \Sigmam_{\hat{\zv}_{\Kc}} &=  \begin{bmatrix}  \mathbf{E}[\hat{\zv}_1\hat{\zv}_1^H]& \mathbf{E}[\hat{\zv}_1\hat{\zv}_2^H]\\\ \mathbf{E}[\hat{\zv}_2\hat{\zv}_1^H] & \mathbf{E}[\hat{\zv}_2\hat{\zv}_2^H]\end{bmatrix} \nonumber \\
   &= \small \begin{bmatrix}  \Id_M +  \mathbf{E}[\hat{\nv}_1 \hat{\nv}_1^H] + \Sigmam_{\hat{\wv}_1} & \Id_M \\\ \Id_M & \Id_M + \mathbf{E}[\hat{\nv}_2 \hat{\nv}_2^H] + \Sigmam_{\hat{\wv}_2}\end{bmatrix}.
\end{align}
Based on \eqref{eq:sigma_zk} the Hadamard's inequality, we have 
\begin{align}\label{eq:det_Zk}
    \det ~\Sigmam_{\hat{\zv}_{\Kc}}   &\leq \det \begin{bmatrix}  \Id_M +  \tilde{\Am}_1^{(1)} + \Sigmam_{\hat{\wv}_1} & 0 \\ 0 & \Id_M +  \tilde{\Am}_2^{(1)}  + \Sigmam_{\hat{\wv}_2}\end{bmatrix} \nonumber \\
    &= \det~ \Sigmam_{\hat{\zv}_{\Kc}^{\rm g} | \tilde{\Am}_{\Kc}^{(1)}}.
\end{align}
Therefore, combining \eqref{eq:h_zk} and \eqref{eq:det_Zk}, an upper bound to $h(\hat{\zv}_{\Kc}| \tilde{\Am}_{\cal K}^{(1)})$ is given by
\begin{align} h(\hat{\zv}_{\cal K}| \tilde{\Am}_{\cal K}^{(1)}) \leq h(\hat{\zv}_{\cal K}^{\rm g}|\tilde{\Am}_{\cal K}^{(1)}),
\end{align}
and we thus have 
\begin{align}\label{T_12}
I(\hat{\xv}_{\Kc}; \hat{\zv}_{\Kc}| \tilde{\Am}_{\cal K}^{(1)}) \leq I(\hat{\xv}_{\Kc}^{\rm g}; \hat{\zv}_{\Kc}^{\text g}| \tilde{\Am}_{\cal K}^{(1)}).
\end{align}
When $\Tc = \{1\}$, i.e., only $\hat{\zv}_1$ is selected, since Gaussian distribution maximizes the differential entropy with the same variance, then we have an upper bound to $h(\hat{\zv}_{1}| \hat{\zv}_{2}, \tilde{\Am}_{\cal K}^{(1)})$ as
\begin{align}\label{eq:h_z1_z2}
    h(\hat{\zv}_{1}| \hat{\zv}_{2}, \tilde{\Am}_{\cal K}^{(1)}) &\leq \log\left((\pi e)^M \det \left(\Sigmam_{\hat{\zv}_1|\hat{\zv}_2}\right)\right).
\end{align}
Assuming that the second order of noises, $\mathbf{E}[\hat{\nv}_k \hat{\nv}_k^H], ~\forall~ k \in \Kc$ are known, $\hat{\zv}_1$ and $\hat{\zv}_2$ are joint Gaussian distribution. According to \cite{eaton2007multivariate}, the covariance matrix of conditional Gaussian distribution $\Sigmam_{\hat{\zv}_1|\hat{\zv}_2}$ is given by 
\begin{align}
    \Sigmam_{\hat{\zv}_1|\hat{\zv}_2} &= \mathbf{E}[{\hat{\zv}_1}{\hat{\zv}_1}^H] - \mathbf{E}[\hat{\zv}_1 \hat{\zv}_2^H] \left(\mathbf{E}[\hat{\zv}_2 \hat{\zv}_2^H]\right)^{-1}\mathbf{E}[\hat{\zv}_2 \hat{\zv}_1^H] \nonumber \\
    &= \Id_M +  \mathbf{E}[\hat{\nv}_1 \hat{\nv}_1^H] + \Sigmam_{\hat{\wv}_1} \nonumber \\
    &- \left(\Id_M + \mathbf{E}[\hat{\nv}_2 \hat{\nv}_2^H] + \Sigmam_{\hat{\wv}_2} \right)^{-1} \nonumber \\
    &= \Om \left(\Id_M +   \mathbf{E}[\hat{\nv}_2 \hat{\nv}_2^H] + \Sigmam_{\hat{\wv}_2} \right)^{-1},
\end{align}
where $\Om = \small \left(\Id_M +  \mathbf{E}[\hat{\nv}_1 \hat{\nv}_1^H] + \Sigmam_{\hat{\wv}_1}\right)\left(\Id_M +   \mathbf{E}[\hat{\nv}_2 \hat{\nv}_2^H] + \Sigmam_{\hat{\wv}_2} \right) - \Id_M$.
Using the fact that $\det(\Am \Bm) = \det(\Am) \det(\Bm)$ and based on Hadamard's inequality, we have 
\begin{align}\label{eq:det_sigma_z1_z2}
    &\det (\Sigmam_{\hat{\zv}_1|\hat{\zv}_2}) \nonumber \\
    = &\det (\Om)  \det\left(\Id_M +  \mathbf{E}[\hat{\nv}_2 \hat{\nv}_2^H] + \Sigmam_{\hat{\wv}_2} \right)^{-1} \nonumber \\
    \leq & \det\left(\small{(\Id_M +  \tilde{\Am}_1^{(1)} + \Sigmam_{\hat{\wv}_1})(\Id_M +  \tilde{\Am}_2^{(1)} + \Sigmam_{\hat{\wv}_2}) - \Id_M}\right) \nonumber \\
    \times &\det\left(\Id_M +  \tilde{\Am}_2^{(1)} +  \Sigmam_{\hat{\wv}_2} \right)^{-1} \nonumber \\
    =& \det \left(\Id_M +  \tilde{\Am}_1^{(1)} +  \Sigmam_{\hat{\wv}_1} - (\Id_M +  \tilde{\Am}_2^{(1)} +  \Sigmam_{\hat{\wv}_2})^{-1}\right) \nonumber \\
    =& \mathbf{E}[\hat{\zv}_1^{\rm g} {\hat{\zv}_1}^{{\rm g} H}] - \mathbf{E}[\hat{\zv}_1^{\rm g} {\hat{\zv}_2}^{{\rm g} H}] \left(\mathbf{E}[\hat{\zv}_2^{\rm g} {\hat{\zv}_2}^{{\rm g} H}]\right)^{-1}\mathbf{E}[\hat{\zv}_2^{\rm g} {\hat{\zv}_1}^{{\rm g} H}] \nonumber \\
    =& \det (\Sigmam_{\hat{\zv}_1^{\rm g}|\hat{\zv}_2^{\rm g}, \tilde{\Am}_{\cal K}^{(1)}}).
\end{align}
Therefore, based on \eqref{eq:h_z1_z2} and \eqref{eq:det_sigma_z1_z2}, we have
    \begin{align}
    h(\hat{\zv}_{1}| \hat{\zv}_{2}, \tilde{\Am}_{\cal K}^{(1)})  \leq h( \hat{\zv}_{1}^{\rm g}| \hat{\zv}^{\rm g}_{2}, \tilde{\Am}_{\cal K}^{(1)}),
\end{align}
and thus we have
    \begin{align}\label{T_1}
    I(\hat{\xv}_{1}; \hat{\zv}_{1}| \hat{\zv}_{2}, \tilde{\Am}_{\cal K}^{(1)})  \leq I(\hat{\xv}_{1}^{\rm g}; \hat{\zv}_{1}^{\rm g}| \hat{\zv}^{\rm g}_{2}, \tilde{\Am}_{\cal K}^{(1)}).
\end{align}
In the similar way, we can prove that 
\begin{align}\label{T_2}
    I(\hat{\xv}_{2}; \hat{\zv}_{2}| \hat{\zv}_{1}, \tilde{\Am}_{\cal K}^{(1)})  \leq I(\hat{\xv}_{2}^{\rm g}; \hat{\zv}_{2}^{\rm g}| \hat{\zv}^{\rm g}_{1}, \tilde{\Am}_{\cal K}^{(1)}).
\end{align}
Therefore, based on \eqref{T_0}, \eqref{T_12}, \eqref{T_1} and \eqref{T_2}, we come to the conclusion that
\begin{align}
    I(\hat{\xv}_{\cal T}; \hat{\zv}_{\cal T}| \hat{\zv}_{\overline{{\cal T}}}, \tilde{\Am}_{\cal K}^{(1)}) \leq I(\hat{\xv}_{\cal T}^{\text g}; \hat{\zv}_{\cal T}^{\text g}| \hat{\zv}_{\overline{{\cal T}}}^{\text g}, \tilde{\Am}_{\cal K}^{(1)}),  ~\forall~ \Tc \subseteq  \Kc.
\end{align}
If the bottleneck constraint $I(\hat{\xv}_{\cal T}^{\text g}; \hat{\zv}_{\cal T}^{\text g}| \hat{\zv}_{\overline{{\cal T}}}^{\text g}, \tilde{\Am}_{\cal K}^{(1)}) \leq \sum_{k \in \Tc} C_k,  \forall \Tc \subseteq  \Kc$ is satisfied, then $I(\hat{\xv}_{\cal T}; \hat{\zv}_{\cal T}| \hat{\zv}_{\overline{{\cal T}}}, \tilde{\Am}_{\cal K}^{(1)}) \leq \sum_{k \in \Tc} C_k,  ~\forall~ \Tc \subseteq  \Kc$ can be guaranteed.

Then, we prove \eqref{eq:b} based on \cite[(A52)]{IBxu},
 \begin{align}
     I(\xv; \hat{\zv}_{\Kc} | \tilde{\Am}_{\cal K}^{(1)}) &\overset{(a)}{=} \sum_{i=1}^M I(x_i; \hat{z}_{\Kc, i} | \tilde{\Am}_{\cal K}^{(1)}) + Q \nonumber \\
      &\geq \sum_{i=1}^M I(x_i; \hat{z}_{\Kc, i} | \tilde{\Am}_{\cal K}^{(1)}) \nonumber \\
     &\overset{(b)}{=}\sum_{i=1}^M I(x_i; \hat{z}_{\Kc, i}^{\rm g} | \tilde{\Am}_{\cal K}^{(1)}) \nonumber \\
      &\overset{(c)}{=}I(\xv; \hat{\zv}_{\Kc}^{\rm g} | \tilde{\Am}_{\cal K}^{(1)}), 
 \end{align}
 where $x_i$ is the $i$-th element of $\xv$ and $\hat{z}_{\Kc, i}$ denotes the element in the $i$-th dimension of $\hat{\zv}_{\Kc}$, i.e., $\{\hat{z}_{1, i}, \hat{z}_{2, i}\}$, $(a)$ holds using the chain rule of mutual information where $Q$ is a non-negative value, $(b)$ holds since for a given $\tilde{\Am}_{\cal K}^{(1)}$, both $\hat{z}_{k, i}$ and $\hat{z}_{k, i}^{\rm g}$ follow $\Cc\Nc(0,  1 + \ceil[\big]{a_{k, i}}_{\cal B} + [\Sigmam_{\hat{\wv}_k}]_{i, i})$, and $(c)$ is satisfied since the elements in $\xv$ and $\hat{\zv}_{\Kc}^{\rm g}$ are independent.
\end{document}